\documentclass[10pt]{llncs}
\usepackage{enumerate}
\usepackage{tikz}
\usetikzlibrary{arrows,shapes,trees}
\usepackage{latexsym}
\usepackage{amssymb}
\usepackage[cmex10]{amsmath}
\usepackage{stmaryrd}
\usepackage{enumerate}
\newcounter{myThmCounter}
\spnewtheorem{thmNo}[myThmCounter]{Theorem}{\normalfont\bfseries}{\itshape}
\spnewtheorem{lemNo}[myThmCounter]{Lemma}{\normalfont\bfseries}{\itshape}

\newenvironment{theoremNo}[1]{\setcounter{myThmCounter}{0#1}\addtocounter{myThmCounter}{-1}\begin{thmNo}}{\end{thmNo}} 
\newenvironment{lemmaNo}[1]{\setcounter{myThmCounter}{0#1}\addtocounter{myThmCounter}{-1}\begin{lemNo}}{\end{lemNo}}

\newcommand{\st}{\mathrel{:}}

\newcommand{\KKK}{\mathcal{K}}
\newcommand{\QQ}{\mathbb{Q}}
\newcommand{\RR}{\mathbb{R}}
\newcommand{\ZZ}{\mathbb{Z}}
\newcommand{\until}[3]{#2 \mathrel{\mathbf{U}_{#3}} #1}
\newcommand{\since}[3]{#2 \mathrel{\mathbf{S}_{#3}} #1}
\newcommand{\Fcount}[1]{\mathrel{\mathbf{C}_{#1}}}
\newcommand{\Pcount}[1]{\mathrel{\mathbf{\overline{C}}_{#1}}}
\newcommand{\fDia}{\Diamond}
\newcommand{\pDia}{\rlap{\hspace{0.24em}\raisebox{1pt}{-}}\Diamond}
\newcommand{\fBox}{\Box}
\newcommand{\pBox}{\rlap{\hspace{0.21em}\raisebox{1pt}{-}}\Box}
\newcommand{\true}{\mathbf{true}}
\newcommand{\false}{\mathbf{false}}

\newcommand{\ltl}{\textrm{LTL}}
\newcommand{\mtl}{\textrm{MTL}}
\newcommand{\fo}{\textrm{FO}}

\newcommand{\MLO}{$\fo(<)$}
\newcommand{\domain}{\mathbb{R}}

\newcommand{\movetoapp}[1]{}

\begin{document}
\title{When is Metric Temporal Logic\\Expressively Complete?}

\author{Paul Hunter}
\institute{Universit\'e Libre de Bruxelles (ULB), Belgium\\\email{paul.hunter@cs.ox.ac.uk}}

\maketitle

\begin{abstract}
A seminal result of Kamp is that over the reals Linear Temporal Logic (\ltl) has
the same expressive power as first-order logic with binary order relation 
$<$ and monadic predicates.  A key question is whether there exists an analogue
of Kamp's theorem for Metric Temporal Logic (\mtl) -- a generalization of \ltl\ in which the Until and Since modalities are annotated with
intervals that express metric constraints. 
Hirshfeld and Rabinovich gave a negative answer, showing that first-order logic with binary
order relation $<$ and unary function $+1$ is strictly more expressive
than MTL with integer constants. However, a recent result of Hunter, 
Ouaknine and Worrell shows that with rational timing constants, 
MTL has the same expressive power as first-order logic, giving a positive answer.
In this paper we generalize these results by giving a precise characterization of those sets of constants for which MTL 
and first-order logic have the same expressive power.  
We also show that full first-order expressiveness can be recovered with the addition of counting modalities,
strongly supporting the assertion of Hirshfeld and Rabinovich that Q2MLO is one of the most expressive decidable fragments of $\fo(<,+1)$.
\end{abstract}

%%%%%%%%%%%%%%%%%%%%%%%%%
\section{Introduction}
%%%%%%%%%%%%%%%%%%%%%%%%%
One of the best-known and most widely studied logics in specification
and verification is \emph{Linear Temporal Logic (\ltl)}: temporal logic
with the modalities \emph{Until} and \emph{Since}.  For discrete-time
systems one considers interpretations of \ltl\ over the integers
$(\mathbb{Z},<)$, and for continuous-time systems one considers
interpretations over the reals $(\mathbb{R},<)$.  A celebrated result
of Kamp~\cite{Kamp68} is that, over both $(\mathbb{Z},<)$ and
$(\mathbb{R},<)$, \ltl\ has the same expressiveness as the \emph{Monadic
  Logic of Order (\MLO)}: first-order logic with binary order relation
$<$ and uninterpreted monadic predicates.  Thus we can benefit from
the appealing variable-free syntax and elementary decision procedures
of \ltl, while retaining the expressiveness and canonicity of
first-order logic.

Over the reals {\MLO} cannot express quantitative properties, such as,
``every request is followed by a response within one time unit''.
This motivates the introduction of \emph{Monadic Logic of Order and
  Metric $\fo_\KKK$}, which augments {\MLO} with a family of unary
function symbols $+c$, $c \in \KKK$ where $\KKK \subseteq \RR$ is some set
of timing constants.
Common choices for $\KKK$ are $\ZZ$, $\{1\}$ (equivalent to $\ZZ$) and $\QQ$, 
however sets of constants such as $\{1,\sqrt{2}\}$ or $\RR$ have practical application
in the specification of systems with two or more timing devices which are
initially synchronized but have independent unit time length. 
We observe that
with simple arithmetic any integer linear combination of elements in $\KKK$ can be derived as a unary function, thus
we restrict our attention to sets that are closed under integer linear combinations, that is,
additive subgroups of $\RR$.

There have been a variety of proposals for quantitative temporal logics, with
modalities definable in {$\fo_\KKK$}\@ (see,
e.g.,~\cite{AH-survey,AH93,AH94,Hen98,HRS98,HR04,HOW13}).  
Typically these temporal logics can be seen as quantitative
extensions of \ltl\@.  However, until~\cite{HOW13} there was
no fully satisfactory counterpart to Kamp's theorem in the quantitative
setting.

The best-known quantitative temporal logic is \emph{Metric Temporal
  Logic ($\mtl$)}, introduced over 20 years ago in~\cite{koymans}.  $\mtl$
arises by annotating the temporal modalities of LTL with real intervals 
representing metric constraints.  It is usual to restrict the endpoints of 
the intervals to some $\KKK \subseteq \RR$, and as we are interested in 
various choices of $\KKK$ we denote this as $\mtl_\KKK$.
Since the $\mtl_\KKK$ operators are definable in {$\fo_\KKK$}, it is immediate that
one can translate $\mtl_\KKK$ into {$\fo_\KKK$}\@.  The main question addressed by this
paper is when does the converse apply?

Several previous results, illustrating that the question is non-trivial, can be succinctly specified with our notation: 
\begin{itemize}
\item Kamp~\cite{Kamp68}: $\mtl_{\{0\}} = \ltl = \textrm{\MLO} = \fo_{\{0\}}$.
\item Hirshfeld and Rabinovich~\cite{HR07}: $\mtl_\ZZ \neq \fo_{\{1\}} = \fo_\ZZ$.
\item Hunter, Ouaknine and Worrell~\cite{HOW13}: $\mtl_\QQ = \fo_\QQ$.
\end{itemize}

The first main result of this paper generalizes these results by giving a precise characterization of when
$\mtl_\KKK$ is expressively complete.

\begin{theorem}\label{thm:main1}
Let $\KKK$ be an additive subgroup of $\RR$.  Then $\mtl_\KKK = \fo_\KKK$ if and only if $\KKK$ is dense.
\end{theorem}

Two consequences of this theorem are that $\mtl_\RR$ is expressively complete (for $\fo_\RR$), and, in contrast to $\mtl_\ZZ \neq \fo_{\{1\}}$, $\mtl$ with interval endpoints taken from $\ZZ[\sqrt{2}] = \{a+b\sqrt{2}\st a,b \in \ZZ\}$ is able to express all of $\fo_{\{1,\sqrt{2}\}}$.

It also follows from our proof of Theorem~\ref{thm:main1} and the result of~\cite{HR07} that if $\mtl_K \neq \fo_K$ then even with 
a (possibly infinite) set of arbitrary additional modal operators of bounded quantifier depth the inequality remains.   
Examples of separating formula are, for sufficiently large $n$,  the modal operator $\Fcount{n}\varphi$
which asserts that $\varphi$ occurs at least $n$ times in the next time interval and its temporal dual $\Pcount{n} \varphi$.
Our second main result is to show that for expressive completeness it is sufficient to add the (infinite) set of these counting operators.
That is, if we define $\mtl_\ZZ\textrm{+C}$ as the logic of $\mtl_\ZZ$ with the additional operators $\{\Fcount{n}, \Pcount{n}\st n \in \mathbb{N}\}$, then
\begin{theorem}\label{thm:main2}
$\mtl_\ZZ\textrm{+C}$ has the same expressive power as $\fo_\ZZ$.
\end{theorem}
In~\cite{HR04} Hirshfeld and Rabinovich considered the addition of counting modalities to MITL: Metric Temporal Logic without singleton (punctual) intervals.  
They showed the resulting logic had the same expressive power as Q2MLO, a decidable fragment of $\fo_{\{1\}}$.  Our result supports their claim that Q2MLO
is one of the most expressive decidable fragments of $\fo_{\{1\}}$: by adding the operators $\fDia_{\{1\}}X$ ($X$ occurs in exactly one time unit) and
$\pDia_{\{1\}}X$ ($X$ occurred exactly one time unit ago) the resulting logic has the full expressive power of $\fo_{\{1\}}$.

%The paper is organized as follows.  In Section~\ref{sec:prelim} we introduce the definitions used throughout.  We prove our main results in Sections~\ref{sec:char} and~\ref{sec:count}, and we conclude in Section~\ref{sec:conc} with conclusions and further work.

%%%%%%%%%%%%%%%%%%%%%%%%%%%%
\section{Preliminaries}\label{sec:prelim}
%%%%%%%%%%%%%%%%%%%%%%%%%%%%
In this section we define the concepts and notation used throughout the paper.

We say $\KKK \subseteq \RR$ is \emph{dense} if for all $a < b \in \KKK$, there exists $c \in \KKK$ such that $a < c < b$. 
In the following, $\KKK \subseteq \RR$ is an additive subgroup of $\RR$.

\subsubsection*{First-order logic.}
Formulas of \emph{Monadic Logic of Order and Metric with constants $\KKK$ ($\fo_\KKK$)} are
first-order formulas over a signature with a binary relation symbol
$<$, an infinite collection of unary predicate symbols
$P_1,P_2,\ldots$, and a (possibly infinite) family of unary function symbols
$+c$, $c \in \KKK$.  Formally, the terms of {$\fo_\KKK$} are
generated by the grammar $t::=x\mid t+c$, where $x$ is a variable and
$c \in \KKK$.  Formulas of {$\fo_\KKK$} are given by the
following syntax:
\[ \varphi ::= \true \mid P_i(t) \mid t < t \mid \varphi \wedge
\varphi \mid \neg \varphi \mid \exists x \, \varphi \, , \] where $x$
denotes a variable and $t$ a term.  
%Without loss of generality it
%suffices to consider only terms of the form $x+q$, $q\in \mathbb{Q}$.

We consider interpretations of {$\fo_\KKK$} over the real line, $\RR$, with the natural
interpretations of $<$ and $+c$.  It follows that a structure for
{$\fo_\KKK$} is determined by an interpretation of the monadic predicates.

%We say a $\fo_\KKK$ formula with one free variable $\varphi(x)$, is \emph{bounded} if there is some $N\in \KKK$
%such that all variables are constrained to lie in the interval $(x-N,x+N)$.
Given terms $t_2$ and $t_2$, we define $\mathrm{Bet}_\KKK(t_1,t_2)$ to consist
of the {$\fo_\KKK$} formulas in which
\begin{enumerate}[(i)]
\item each subformula $\exists z\, \psi$ has the form
$\exists z\,((t_1 < z<t_2) \wedge \chi)$,
  i.e., each quantifier is relativized to the open interval between
  $t_1$ and $t_2$;
\item in each atomic subformula $P(t)$ the term
  $t$ is a bound occurrence of a variable.
\end{enumerate}

Clauses (i) and (ii) ensure that a formula in $\mathrm{Bet}_\KKK(t_1,t_2)$
only refers to the values of monadic predicates on points in the open
interval $(t_1,t_2)$.  We say that a formula $\varphi(x)$ in
$\mathrm{Bet}_\KKK(x-N,x+N)$ is \emph{$N$-bounded}.

\subsubsection*{Metric Temporal Logic.}
Given a set $\boldsymbol{P}$ of atomic propositions, the formulas of
\emph{Metric Temporal Logic with constants $\KKK$ ($\mtl_\KKK$)} are built from $\boldsymbol{P}$
using boolean connectives and time-constrained versions of the
\emph{Until} and \emph{Since} operators $\until{}{}{}$ and
$\since{}{}{}$ as follows:
\[ \varphi ::= \true\mid P \mid \varphi \wedge \varphi \mid \neg \varphi \mid
\until{\varphi}{\varphi}{I} \mid \since{\varphi}{\varphi}{I} \, , \]
where $P \in \boldsymbol{P}$ and $I\subseteq (0,\infty)$ is an
interval with endpoints in $\KKK_{\geq 0}\cup\{\infty\}$.

Intuitively, the meaning of $\until{\varphi_2}{\varphi_1}{I}$ is that
$\varphi_2$ will hold at some time in the interval $I$, and until then
$\varphi_1$ holds. 
More precisely, the semantics of \(\mtl_\KKK\) are defined as follows.  A
\emph{signal} is a function $f:\domain \to 2^{\boldsymbol{P}}$.  Given a signal $f$
and $r \in \domain$, we define the satisfaction relation $f,r \models
\varphi$ by induction over $\varphi$ as follows:
\begin{itemize}
\item $f,r \models p$ iff $p \in f(r)$,
\item $f,r \models \neg \varphi$ iff $f,r \not\models \varphi$,
\item $f,r \models \varphi_1 \wedge \varphi_2$ iff $f,r \models \varphi_1$ and $f,r \models \varphi_2$,
\item $f,r \models \until{\varphi_2}{\varphi_1}{I}$ iff there exists
  $t>r$ such that $t-r \in I$, $f,t \models \varphi_2$ and $f,u \models
  \varphi_1$ for all $u\in (r,t)$,
\item $f,r \models \since{\varphi_2}{\varphi_1}{I}$ iff there exists
  $t<r$ such that $r-t \in I$, $f,t \models \varphi_2$ and $f,u
  \models \varphi_1$ for all $u \in (t,r)$.
\end{itemize}

LTL can be seen as a restriction of MTL with only the interval
$I=(0,\infty)$, so in particular $\mathrm{LTL} = \mtl_{\{0\}}$.  
%Indeed, if $I=(0,\infty)$ then we omit the annotation
%$I$ in the corresponding temporal operator since the constraint is
%vacuous.  
MITL is a restriction of $\mtl_\ZZ$ where singleton intervals, that is intervals of the form $\{c\}$, 
do not occur in the $\until{}{}{}$ and $\since{}{}{}$ operators.

%We use arithmetic expressions to denote
%intervals.  For example, we write $\until{}{}{<3}$ for
%$\until{}{}{(0,3)}$ and $\until{}{}{=1}$ for $\until{}{}{\{1\}}$.  

We say the $\until{}{}{I}$ and $\since{}{}{I}$ operators are
\emph{bounded} if $I$ is bounded, otherwise we say that the operators
are \emph{unbounded}.

We introduce the derived connectives $\fDia_I \varphi:=
\until{\varphi}{\true}{I}$ ($\varphi$ will be true at some point in
interval $I$) and $\pDia_I \varphi :=\since{\varphi}{\true}{I}$
($\varphi$ was true at some point in interval $I$ in the past).  We
also have the dual connectives $\fBox_I \varphi := \neg \fDia_I \neg
\varphi$ ($\varphi$ will hold at all times in interval $I$ in the
future) and $\pBox_I := \neg \pDia_I \neg \varphi$ ($\varphi$ was true
at all times in interval $I$ in the past).

\paragraph{Counting modalities. }
The counting modalities $\Fcount{n}\varphi$ and $\Pcount{n}\varphi$ are defined for all $n \in \mathbb{N}$
and are interpreted as $\varphi$ will be true for at least $n$ distinct occasions in the next/previous time unit.
That is, for any signal $f$ and $r \in \RR$:
\begin{itemize}
\item $f,r \models \Fcount{n}\varphi$ iff there exists $r_1 < \cdots < r_n \in (r,r+1)$ with $f,r_i \models \varphi$ for all $i$.
\item $f,r \models \Pcount{n}\varphi$ iff there exists $r_1 < \cdots < r_n \in (r-1,r)$ with $f,r_i \models \varphi$ for all $i$.
\end{itemize}
We define \emph{$\mtl_\KKK$ with counting ($\mtl_\KKK\textrm{+C}$)} to be the extension of $\mtl_\KKK$ by the operations $\{\Fcount{n}, \Pcount{n} \st n \in \mathbb{N}\}$. 

\subsubsection*{Expressive Equivalence.}
\label{sec:exp-equiv}
Given a set $\boldsymbol{P}=\{P_1, \ldots, P_m\}$ of monadic
predicates, a signal $f:\domain \to 2^{\boldsymbol{P}}$ defines an
interpretation of each $P_i$, where $P_i(r)$ iff $P_i \in
f(r)$.  As observed earlier, this is sufficient to define the
model-theoretic semantics of {$\fo_\KKK$}, enabling us to relate the
semantics of {$\fo_\KKK$} and $\mtl_\KKK$.

Let $\varphi(x)$ be an {$\fo_\KKK$} formula with one free variable and $\psi$
an $\mtl_\KKK$ formula.  We say $\varphi$ and $\psi$ are \emph{equivalent} if
for all signals $f$ and $r \in \domain$:
\[ f \models \varphi[r] \Longleftrightarrow f,r \models \psi.\]
We say $\mtl_\KKK$ and $\fo_\KKK$ have the same expressive power, written $\mtl_\KKK = \fo_\KKK$,
if for all formulas with one free variable $\varphi(x) \in \fo_\KKK$ there is an equivalent formula $\varphi^\dag \in \mtl_\KKK$ and vice versa.

%%%%%%%%%%%%%%%%%%%%%%%%%
\section{Characterization of expressively complete $\mtl$}\label{sec:char}
%%%%%%%%%%%%%%%%%%%%%%%%%
The goal of this section is to prove:
\begin{theoremNo}{\ref{thm:main1}}
Let $\KKK$ be an additive subgroup of $\RR$.  Then $\mtl_\KKK = \fo_\KKK$ if and only if $\KKK$ is dense.
\end{theoremNo}

First we consider the ``only if'' direction.  Central to this is the following easily proven result:
\begin{lemma}\label{lem:nondense}
Let $\KKK$ be an additive subgroup of $\RR$.  If $\KKK$ is not dense then $\KKK = \epsilon \ZZ$ for some $\epsilon > 0$.
\end{lemma}

It now follows by a simple scaling argument and the result $\mtl_\ZZ \neq \fo_\ZZ$~\cite{HR07}  that if $\KKK$ is not dense then $\mtl_\KKK \neq \fo_\KKK$.  We refer the reader to the appendix for details.

In fact~\cite{HR07} showed a much stronger result: even with (possibly infinite) additional arbitrary modal operators of bounded quantifier depth $\mtl_\ZZ$ cannot fully express $\fo_\ZZ$.  This result clearly carries over to $\KKK = \epsilon \ZZ$, thus in the non-dense case $\mtl_\KKK$ is ``quite far'' from $\fo_\KKK$.
\begin{corollary}
Let $\KKK$ be a non-dense additive subgroup of $\RR$.   With additional arbitrary modal operators of bounded quantifier depth $\mtl_\KKK$ cannot fully express $\fo_\KKK$.
\end{corollary}

Returning to the ``if'' direction in the proof of Theorem~\ref{thm:main1}, we focus on the non-trivial case ($\KKK$ infinite), as the trivial case $\KKK = \{0\}$ is covered by Kamp's theorem~\cite{Kamp68}.
Our strategy parallels the proof of expressive completeness of $\mtl_\QQ$ in~\cite{HOW13}:  We first show expressive completeness for bounded formulas, and then, using a refinement of syntactic separation~\cite{G81,HOW13}, extend this to all $\fo_\KKK$ formulas.
%\begin{enumerate}[Step 1. ]
%\item We first show $\mtl_\KKK$ syntactically separable for all additive subgroups $\KKK \subseteq \RR$.  Our overall proof as well as this result requires a slight refinement of the definition of syntactically separable.This allows us to focus on bounded formulas.
%\item We next show that on bounded formulas the unary $+c$ functions can be removed from all variables but the free one.  The resulting formula can then be decomposed into $\fo(<)$ formulas in $\mathrm{Bet}(x+c,x+c')$ for various constants $c,c' \in \KKK$.
%\item Using similar ideas to~\cite{HOW13} we show how $\mtl_\KKK$ for dense $\KKK$ is able to express formulas in $\mathrm{Bet}(x+c,x+c')$, giving expressive completeness on bounded formulas.
%\item Finally we use the metric separation of $\mtl_\KKK$ to extend this expressiveness to all $\fo_\KKK$ formulas.
%\end{enumerate}

\subsection{Expressive completeness for bounded formulas}
To show that bounded $\fo_\KKK$ formulas can be expressed by $\mtl_\KKK$ we proceed in a similar manner to~\cite{HOW13}.
\begin{enumerate}[Step 1. ]
\item We first remove any occurrence of a unary $+c$ function applied to a bound variable.
\item Using a composition argument (see e.g.~\cite{GPSS80,HR06}) we then reduce the problem to showing expressive completeness for formulas in $\mathrm{Bet}_{\{0\}}(x,x+c)$.
\item Exploiting a normal form of~\cite{GPSS80} and the denseness of $\KKK$ we show how an $\mtl_\KKK$ formula can express any formula in $\mathrm{Bet}_{\{0\}}(x,x+c)$, and hence any bounded formula.
\end{enumerate}
Our proof differs significantly to that of~\cite{HOW13} notably at Steps 1 and 2.  In~\cite{HOW13} the authors were able to scale $\fo_\QQ$ formulas to $\fo_{\{1\}}$ and then use the regularity of the integers to reduce the problem to formulas in $\mathrm{Bet}_{\{0\}}(x,x+1)$ (so-called \emph{unit-formulas}).  For more general $\KKK$ however neither of these steps are applicable so instead we introduce a normal form for $\fo_\KKK$ formulas which simplifies the removal of the unary functions.

\subsubsection*{Step 1. Removing unary functions.}
Given an $N$-bounded $\fo_\KKK$ formula with one free variable $x$, we show that it is equivalent to a $N'$-bounded formula (over a possibly larger set of monadic predicates, suitably interpreted) in which the unary functions are only applied to $x$.  We can remove occurrences of unary functions within the scope of monadic predicates by introducing new predicates.  That is, we replace $P(y+c)$ with $P^c(y)$, the intended interpretation of $P^c$ being $\{r \st r+c \in P\}$.  We will later replace $P^c(y)$ with $\fDia_{\{c\}} P$ when completing the translation to $\mtl_\KKK$.  Thus it suffices to demonstrate how to remove the unary functions from the scope of the $<$ operator.  For this we introduce a normal form where all inequality constraints are replaced with interval inclusions and the intervals satisfy the following hierarchical condition: if $y$ is quantified to $(x+c,z+c')$ then all intervals involving $y$ and a variable that was free when $y$ was quantified are affine translations of $(x+c,y)$ or $(y,z+c')$.  We note that the results of this section apply for any additive subgroup $\KKK \subseteq \RR$.  

\begin{definition}
An \emph{interval-guarded formula} is a $\fo_\KKK$-formula such that all quantifiers are of the form $\exists x \in (y+c,y'+c')$ where $y,y'$ are free variables and $c,c' \in \KKK$.
A \emph{Hierarchical Interval Formula (HIF)} is an interval-guarded $\fo_\KKK$-formula defined inductively as follows.
\begin{itemize}
\item Any $<$-free, quantifier-free $\fo_\KKK$-formula is a HIF;
\item If $\varphi_1, \varphi_2$ are HIFs then so are $\neg \varphi_1$ and $\varphi_1 \vee \varphi_2$; and
\item If $\varphi(\overline{x},y)$ is a HIF and there exists $x_l,x_r \in \overline{x}$ and $c_l,c_r \in \mathcal{K}$ such that
the only intervals in $\varphi$ involving $y$ and a free variable are of the form $(x_l+c_l+c,y+c)$ or $(y+c,x_r+c_r+c)$ for some $c \in \mathcal{K}$, then 
\(\exists y \in (x_l+c_l,x_r+c_r).\varphi(\overline{x},y)\) is a HIF.
\end{itemize}
\end{definition}

For space reasons we omit the proof that HIFs are a normal form for $N$-bounded $\fo_\KKK$ formulas with one free variable.  The full details can be found in the appendix.

\begin{lemma}\label{lem:HIF2}
Every $N$-bounded $\fo_\KKK$ formula with one free variable is equivalent to a HIF.
\end{lemma}
%\begin{proof}[Proof sketch.]
% We first transform the $\fo_\KKK$ formula into an interval-guarded formula by taking a disjunction over all maximally-constraining terms.
%\end{proof}

The final stage of this step is to remove the application of unary functions to all bound variables. 
\begin{lemma}\label{lem:remove}
Let $\KKK$ be an additive subgroup of $\RR$ and $\varphi(x)$ be an $N$-bounded $\fo_\KKK$ formula with one free variable.  Then $\varphi(x)$ is equivalent to an $N'$-bounded $\fo_\KKK$ formula $\varphi'(x)$ in which the unary functions are only applied to $x$.
\end{lemma}
\begin{proof}
Let us say there is a \emph{violation} if a unary function is applied to a variable other than $x$.  
Following Lemma~\ref{lem:HIF2} and the comments at the start of the section it suffices to consider HIFs and remove all violations from intervals.  We proceed from any maximal subformula of $\varphi(x)$, $\psi(x, \overline{y}) = \exists z \in (s,t).\theta(x,\overline{y},z)$ where there is a violation, say $t = y_j+c$ (the case for $s = y_j+c$ being similar).  Consider $\psi' = \exists z' \in (s-c,t-c).\theta(x,\overline{y},z'+c)$.  $\psi'$ is clearly equivalent to $\psi$ and is $(N+c)$-bounded.  It suffices to show that $s-c$ is not a violation as this implies all violations in $\psi'$ occur in proper subformulas and the result then follows by induction. The critical case is if $s = y_k+c'$. Then, as $\varphi$ is a HIF and $y_j$ and $y_k$ are bound in $\varphi$, it follows that $j \neq k$.  Suppose $j<k$. Then $y_j + c- c'$ must have been an endpoint on the interval constraining $y_k$ at the point where $y_k$ was quantified.  As $\psi$ is maximal, it follows that $c = c'$.  Likewise if $k<j$.  Therefore $s-c$ is not a violation.
\end{proof}

\subsubsection*{Step 2. Reduction to  $\mathrm{Bet}_{\{0\}}(x,x+c)$ formulas.}
Suppose now $\varphi(x)$ is an $N$-bounded $\fo_\KKK$ formula in which the unary functions are only applied to $x$.  
Let $c_0 < c_1 < \ldots < c_n$ be the constants in $\KKK$ (including $0$) corresponding to the unary functions that are applied to $x$.
Let $\varphi'(\overline{z})$ be the formula resulting from replacing each term $x+c_i$ with a new variable $z_i$.  Then $\varphi(x)$ is equivalent to $\exists \overline{z}. (z_0 < \cdots < z_n) \wedge  \varphi'(\overline{z}) \wedge \bigwedge (z_i = x+c_i)$.  Moreover, $\varphi'$ does not contain any unary functions and is thus a formula of $\fo_{\{0\}}$.  A standard model-theoretic argument (see~\cite{Kamp68,GPSS80,HR06}) shows that $(z_0 < \cdots < z_n) \wedge \varphi'(\overline{z})$ can be written as a finite disjunction of formulas of the form $\bigwedge_{i=0}^n \psi_i(z_i) \wedge \bigwedge_{i=0}^{n-1} \chi_i(z_i,z_{i+1})$ where each $\psi_i$ is a boolean combination of monadic predicates and each $\chi_i \in \mathrm{Bet}_{\{0\}}(z_i,z_{i+1})$.%\footnote{As $\varphi(x)$ is $N$-bounded we do not need to consider the types for the intervals $(-\infty,x_0)$ and $(x_n,\infty)$.}  
Thus $\varphi(x)$ can be written as a finite disjunction of formulas of the form
\[\bigwedge_{i=0}^n \psi_i(x+c_i) \wedge \bigwedge_{i=0}^{n-1} \chi_i(x+c_i,x+c_{i+1}).\]
Now $\psi_i(x+c_i)$ is clearly expressible by the $\mtl_\KKK$ formula $\fDia_{\{c_i\}} \psi_i^\dag$, where $\psi_i^\dag$ is the obvious translation of $\psi_i(x)$ to $\mtl_\KKK$.  Likewise, if $\chi_i^\dag$ were an $\mtl_\KKK$ formula expressing $\chi_i(x,x+c_{i+1}-c_i)$ then $\fDia_{\{c_i\}} \chi_i^\dag$ would be an $\mtl_\KKK$ formula expressing $\chi_i(x+c_i,x+c_{i+1})$.  Thus we have reduced the problem of expressing $N$-bounded $\fo_\KKK$ formulas to expressing every formula in $\mathrm{Bet}_{\{0\}}(x,x+c)$.

\subsubsection*{Step 3. Expressive completeness for bounded formulas.}
Critical to this step is the following
definition and lemma from~\cite{GPSS80}.

A \emph{decomposition formula} $\delta(x,y)$ is any formula of the
form
\begin{align*}
x<y & \wedge \exists z_0 \ldots \exists z_n\,(x=z_0< \cdots < z_n=y) \\
    & \wedge \bigwedge\{ \varphi_i(z_i) : 0 < i < n\}\\
    & \wedge \bigwedge\{ \forall u\,((z_{i-1}<u<z_{i}) \rightarrow \psi_i(u)) :
0 < i \leq n\}
\end{align*}
where $\varphi_i$ and $\psi_i$ are \ltl\ formulas regarded as unary
predicates.

\begin{lemma}[\cite{GPSS80}]
Over any domain with a complete linear order, every formula $\psi(x,y)$ in
$\mathrm{Bet}_{\{0\}}(x,y)$ is equivalent to a boolean combination of
decomposition formulas $\delta(x,y)$.
\label{lem:gpss}
\end{lemma}

It follows that it suffices to show $\mtl_\KKK$ is able to express a decomposition formula.  
The proof of this result very closely follows the proof in~\cite{HOW13}, so we only outline the ideas 
and refer the reader to the appendix for the full details.

%Expressing Bet(x,y)
\begin{lemma}\label{lem:decomp}
Any decomposition formula $\delta(x,x+c)$ is equivalent to an $\mtl_\KKK$ formula.
\end{lemma}
\begin{proof}[Sketch]
The proof is by induction on $n$, the number of existential quantifiers in $\delta(x,x+c)$.
We divide the interval $(x,x+c)$
into small intervals of width $\nu \in \KKK$ where $0< \nu \leq \frac{c}{2n}$.  The fact that $\KKK$ is non-trivial and dense guarantees that $\nu$ exists.
We then consider three cases depending on where the witnesses for the 
existential quantifiers of $\delta$ lie (taking a disjunction to cover all cases).
If all witnesses lie in a single interval in the first half of $(x,x+c)$ then we can assert in $\mtl_\KKK$:
$\psi_1$ holds until some point in the interval, then subsequent witness points occur within $\nu$ time units of the previous one.
%If $\nu$ is small enough, this formula makes no assertions beyond the interval $(x,x+c)$.
If instead all witnesses lie in a single interval in the second half of $(x,x+c)$ we assert:
In $c$ time units $\psi_n$ would have held since a point in the interval, and each witness point was preceded within $\nu$ time units by another.
Finally, if there is some $k$ such that $x+k\nu$ separates the witnesses, we divide $\delta(x,x+c)$ into a $\mathrm{Bet}_{\{0\}}(x,x+k\nu)$ formula and 
a $\mathrm{Bet}_{\{0\}}(x+k\nu,x+c)$ formula and apply the inductive hypothesis.
\end{proof}

\noindent Combining Kamp's Theorem and the results of this section yields:

\begin{lemma}\label{lem:bounded}
Let $\KKK$ be a dense additive subgroup of $\RR$.  Any $N$-bounded $\fo_\KKK$ formula with one free variable is equivalent to an $\mtl_\KKK$ formula.
\end{lemma}

\subsection{Syntactic separation of $\mtl_\KKK$}
Having established that $\mtl_\KKK$ can express $N$-bounded $\fo_\KKK$ formulas when $\KKK$ is dense
we now turn to extending the result to all $\fo_\KKK$.  Our results for this section hold for all non-trivial
additive subgroups $\KKK$.

The notion of \emph{separation} was introduced by Gabbay in~\cite{G81} 
where he showed that every $\ltl$ formula can be equivalently
rewritten as a boolean combination of formulas, each of which depends
only on the past, present or future. 
Hunter, Ouaknine and Worrell~\cite{HOW13} extended this idea for the 
metric setting, showing that each $\mtl_\QQ$ formula can be
equivalently rewritten as a boolean combination of formulas, each
of which depends only on the distant past, bounded present, or distant future. 

Here we use a similar approach, however we need to refine the definition
of distant past and distant future in order to use the separation property in Section~\ref{sec:combining}.
This refinement is, however, simple enough that the proof of separability of $\mtl_\QQ$ in~\cite{HOW13}
can largely be used and we need only indicate the two places where adjustments need to be made to account for our more general setting.

Recall from~\cite{HOW13} the inductive definitions of \emph{future-reach} 
$\mathit{fr}:\mtl_\KKK \to \KKK\cup\{\infty\}$ and \emph{past-reach} 
$\mathit{pr}:\mtl_\KKK \to \KKK\cup\{\infty\}$
\begin{itemize}
\item $\mathit{fr}(p) = \mathit{pr}(p)=0$ for all propositions $p$,
\item $\mathit{fr}(\true) = \mathit{pr}(\true)= 0$,
\item $\mathit{fr}(\neg \varphi) = \mathit{fr}(\varphi)$, $\mathit{pr}(\neg \varphi) = \mathit{pr}(\varphi)$,
\item $\mathit{fr}( \varphi \wedge \psi) = \max\{\mathit{fr}(\varphi),\mathit{fr}(\psi)\}$,
\item $\mathit{pr}( \varphi \wedge \psi) = \max\{\mathit{pr}(\varphi),\mathit{pr}(\psi)\}$,
\item If $n=\inf(I)$ and $m=\sup(I)$:
\begin{itemize}
\item $\mathit{fr}(\until{\psi}{\varphi}{I}) = m + \max\{\mathit{fr}(\varphi),\mathit{fr}(\psi)\}$,
\item $\mathit{pr}(\since{\psi}{\varphi}{I}) = m + \max\{\mathit{pr}(\varphi),\mathit{pr}(\psi)\}$,
\item $\mathit{fr}(\since{\psi}{\varphi}{I}) = \max\{\mathit{fr}(\varphi),\mathit{fr}(\psi)-n\}$,
\item $\mathit{pr}(\until{\psi}{\varphi}{I}) = \max\{\mathit{pr}(\varphi),\mathit{pr}(\psi)-n\}$.
\end{itemize}
\end{itemize}

\noindent Our separation result is then:

\begin{lemma}\label{lem:sep}
Let $\KKK$ be a non-trivial additive subgroup of $\RR$.  For any $c \in \KKK_{\geq 0}$, every $\mtl_\KKK$ formula
is equivalent to a boolean combination of:
\begin{itemize}
\item $\fDia_{\{N\}}  \varphi$ where $\mathit{pr}(\varphi)<N-c$,
\item $\pDia_{\{N\}} \varphi$ where $\mathit{fr}(\varphi)< N-c$, and
\item $\varphi$ where all intervals occurring in the temporal operators are bounded.
\end{itemize}
\end{lemma}
\begin{proof}
The proof follows directly from the proof of the separability of $\mtl_\QQ$ in~\cite{HOW13} as only few assumptions were made about the underlying set of constants, which we now address. 

\begin{itemize}
\item For the equivalence defining $K^+$ and $K^-$ as bounded formulas, we instead need to use:
$K^+(\varphi) \: \leftrightarrow\: \neg (\until{\true}{\neg \varphi}{<\nu})$ and  $K^-(\varphi) \: \leftrightarrow\: \neg (\since{\true}{\neg \varphi}{<\nu})$,  
where $\nu \in \KKK$ is such that $\nu>0$.  Note that as $\KKK$ is non-trivial such a $\nu$ exists.

\item In Step 3 (\textit{Completing the separation}) $N$ was chosen so that $N > \mathit{pr}(\theta)+1$.  Now we choose $N\in \KKK$ such that $N > \mathit{pr}(\theta)+c$.  Note that again as $\KKK$ is non-trivial such a choice is always possible.
\end{itemize}
\end{proof}

\subsection{Expressive completeness for $\fo_\KKK$}\label{sec:combining}
We now use Lemmas~\ref{lem:bounded} and~\ref{lem:sep} to complete the proof of Theorem~\ref{thm:main1}.  
Let $\varphi(x)$ be a $\fo_\KKK$ formula.  We prove by induction on the quantifier depth of $\varphi(x)$ that it is equivalent to 
an $\mtl_\KKK$ formula.
\vspace*{-2ex}
\paragraph{Base case.}  
All atoms are of the form $P_i(x)$, $x=x$, $x<x$, $x+c=x$.  We replace
these by $P_i$, $\true$, $\false$, $\false$ respectively and obtain an
$\mtl_\KKK$ formula which is clearly equivalent to $\varphi$.
\vspace*{-2ex}
\paragraph{Inductive case.}  
Without loss of generality we may assume $\varphi = \exists
y. \psi(x,y)$.  We would
like to remove $x$ from $\psi$.  To this end we take a disjunction
over all possible choices for $\gamma:\{P_1(x), \ldots
P_m(x)\}\to\{\true,\false\}$, and use $\gamma$ to determine the value
of $P_i(x)$ in each disjunct via the formula $\theta_\gamma :=
\bigwedge_{i=1}^m (P_i(x) \leftrightarrow \gamma(P_i))$.  Thus we can
equivalently write $\varphi$ in the form
%\begin{gather}
$\bigvee_\gamma \big(\theta_\gamma(x) \wedge \exists y. \psi_\gamma(x,y)\big) \, ,$
%\end{gather}
where the propositions $P_i(x)$ do not appear in the $\psi_\gamma$.

Now in each $\psi_\gamma$, we may assume, after some arithmetic, $x$ appears only in atoms of the form
$x=z$, $x<z$, $x>z$ and $x+c=z$ for some variable $z$.  We next
introduce new monadic propositions $P_=$, $P_<$, $P_>$, and
$F_c$ for all $c$ such that there is an atom $x+c=z$, and replace each of the atoms containing $x$ in $\psi_\gamma$
with the corresponding proposition.  That is, $x=z$ becomes $P_=(z)$,
$x<z$ becomes $P_<(z)$ and so on.  This yields a formula
$\psi'_\gamma(y)$ in which $x$ does not occur, such that
$\psi'_\gamma(y)$ has the same truth value as $\psi_\gamma(x,y)$ 
for suitable interpretations of the new propositions.

By the induction hypothesis, for each $\gamma$ there is an $\mtl_\KKK$
formula $\theta_\gamma^\dag$ equivalent to $\theta_\gamma(x)$, and an
$\mtl_\KKK$ formula $\psi^\dag_\gamma$ equivalent to
$\psi'_\gamma(y)$.  Then our
original formula $\varphi$ has the same truth value at each point $x$ as
\( \varphi' := \bigvee_\gamma \big(\theta_\gamma^\dag \wedge 
(\pDia \psi_\gamma^{\dagger} \vee \psi_\gamma^{\dagger} \vee \fDia 
\psi_\gamma^{\dagger})\big)\)
for suitable interpretations of $P_=$, $P_<$, $P_>$ and the $F_c$.  

Let $c_{\max} \in \KKK$ be the largest, in absolute value, element of $\KKK$ for which the 
propositional variable $F_c$ was introduced.
By Lemma~\ref{lem:sep}, $\varphi'$ is equivalent to a boolean
combination of formulas
\begin{enumerate}[(I) ]
\item $\fDia_{\{N\}}  \theta$ where $\mathit{pr}(\theta)<N-|c_{\max}|$,
\item $\pDia_{\{N\}} \theta$ where $\mathit{fr}(\theta)< N-|c_{\max}|$, and
\item $\theta$ where all intervals occurring in the temporal operators are bounded.
\end{enumerate}
Now in formulas of type (I) above, we know the intended value of each
of the propositional variables $P_=, P_<, P_>$ and $F_c$: they are all
$\false$ except $P_<$, which is $\true$.  So we can replace these
propositional atoms by $\true$ and $\false$ as appropriate and obtain
an equivalent $\mtl_\KKK$ formula which does not mention the new
variables.  Likewise we know the value of each of propositional
variables in formulas of type (II): all are $\false$ except $P_>$,
which is $\true$; so we can again obtain an equivalent $\mtl_\KKK$
formula which does not mention the new variables.  It remains to deal
with each of the bounded formulas, $\theta$.  As $\mtl_\KKK$ is definable in $\fo_\KKK$, 
there exists a formula $\theta^\ast(x)\in \fo_\KKK$, with predicates from $\{P_=, P_<, P_>, F_c\}$,
equivalent to $\theta$.  It is clear that as
$\theta$ is bounded, there is an $N$ such that $\theta^\ast$ is
$N$-bounded.  We now unsubstitute each of the introduced propositional
variables.  That is, replace in $\theta^\ast(x)$ all occurrences of
$P_=(z)$ with $z=x$, all occurrences of $P_<(z)$ with $x<z$ etc.  The
result is an equivalent formula $\theta^+(x) \in \fo_\KKK$, which is still
$N$-bounded as we have not removed any constraints on the variables of
$\theta^\ast$.  From Lemma~\ref{lem:bounded}, it follows that there
exists an $\mtl_\KKK$ formula $\delta$ that is equivalent to
$\theta^+$, i.e., equivalent to $\theta$.

%%%%%%%%%%%%%%%%%%%%%%%%%
\section{Expressive completeness of $\mtl_\ZZ$ with counting}\label{sec:count}
%%%%%%%%%%%%%%%%%%%%%%%%%
In this section we show
\begin{theoremNo}{\ref{thm:main2}}
$\mtl_\ZZ\textrm{+C}$ has the same expressive power as $\fo_\ZZ$.
\end{theoremNo}

In fact we show a slightly stronger result involving an extension of Q2MLO (see~\cite{HR06}) by punctuality quantifiers.

\begin{definition}
\emph{Q2MLO with punctuality (PQ2MLO)} is an extension of $\fo_{\{0\}}$ (and a restriction of $\fo_{\{1\}}$) defined by the following syntax:
%\begin{eqnarray*} 
%\varphi &::=& \true \mid P_i(x) \mid x < y \mid \varphi \wedge
%\varphi \mid \neg \varphi \mid \exists x \, \varphi \mid \\
%&&\qquad\exists_{x}^{x+1} y\, \psi \mid   \exists_{x-1}^{x} y\, \psi\mid   \fDia_{1}^x y.\, \chi\mid   \pDia_{1}^x y.\, \chi\, ,
%\end{eqnarray*}
%\begin{eqnarray*} 
\[\varphi ::= \true \mid P_i(x) \mid x < y \mid \varphi \wedge
\varphi \mid \neg \varphi \mid \exists x \, \varphi \mid \exists_{x}^{x+1} y\, \psi \mid   \exists_{x-1}^{x} y\, \psi\mid   \fDia_{1}^x y.\, \chi\mid   \pDia_{1}^x y.\, \chi\, ,
\]
where $x$ and $y$ denote variables, $\psi$ denotes a PQ2MLO formula with two free variables $x$ and $y$, and $\chi$ denotes a PQ2MLO formula with one free variable, $y$.
\emph{Q2MLO} is the restriction of PQ2MLO to formulas that do not contain the punctual quantifiers $\fDia_{1}^x$ and $\pDia_{1}^x$.
\end{definition}
The quantifiers $\exists_{x}^{x+1} y$, $\exists_{x-1}^{x} y$, $\fDia_{1}^x y$ and $\pDia_{1}^x y$ are interpreted as $\exists y \in (x,x+1)$, $\exists y \in (x-1,x)$, $\exists y.\, (y = x+1)$ and $\exists y.\, (y = x-1)$ respectively.  

\begin{theorem}\label{thm:main2plus} $\fo_\ZZ$, PQ2MLO and $\mtl_\ZZ\textrm{+C}$ all have the same expressive power.%:
%\begin{itemize}
%\item $\fo_\ZZ$,
%\item Q2MLO with punctuality,
%\item $\mtl_\ZZ$ with counting.
%\end{itemize}
\end{theorem}

It is clear that $\fo_\ZZ$ is at least as expressive as the other two.  To show the equivalence of PQ2MLO and $\mtl_\ZZ\textrm{+C}$ we use the following result of~\cite{HR06}.
\begin{theorem}[\cite{HR06}]\label{thm:Q2MLO}
MITL with counting has the same expressive power as Q2MLO.
\end{theorem}

We also observe that if $\varphi(y)$ is a formula of PQ2MLO that is equivalent to $\varphi' \in \mtl_\ZZ\textrm{+C}$ then $\fDia_{1}^x y \varphi(y) $ is equivalent to $\fDia_{\{1\}} \varphi'$ and $\pDia_{1}^x y \varphi(y) $ is equivalent to $\pDia_{\{1\}} \varphi'$.  The result then follows by induction on the nesting depth of the punctual operators ($\fDia_{1}^x y$ / $\pDia_{1}^x y$ and $\fDia_{\{1\}}$ / $\pDia_{\{1\}}$) and Theorem~\ref{thm:Q2MLO}.
%Also $\fDia_{\{n\}} \varphi = \fDia_{\{1\}}\fDia_{\{1\}}\cdots \fDia_{\{1\}} \varphi$

It remains to show any formula in $\fo_\ZZ$ has an equivalent PQ2MLO formula.  From the proof of Theorem~\ref{thm:main1} in the previous section, it is sufficient to derive an analogue of Lemma~\ref{lem:decomp} for $\fo_{\{1\}}$.  That is, we need only consider $\fo_{\{1\}}$ formulas of the form $\delta(x) = \delta(x,x+1)$ where:
\begin{align*}
\delta(x,y) &=  \exists z_0 \ldots \exists z_n\,(x=z_0< \cdots < z_n=y) \\
    & \quad \wedge \bigwedge\{ \varphi_i(z_i) : 0 < i < n\}\\
    & \quad \wedge \bigwedge\{ \forall u\,((z_{i-1}<u<z_{i}) \rightarrow \psi_i(u)) :
0 < i \leq n\}.
\end{align*}

\noindent For $1 \leq j \leq 2n-1$ let
\begin{align*}
\delta_{j}(x,y) &= \exists z_0 \ldots \exists z_k\,(x=z_0< \cdots < z_k=y) \\
    & \quad \wedge \bigwedge\{ \varphi_i(z_i) : 0 < i \leq \lfloor \frac{j}{2} \rfloor\}\\
    & \quad \wedge \bigwedge\{ \forall u\,((z_{i-1}<u<z_{i}) \rightarrow \psi_i(u)) :
0 < i \leq k\},
\end{align*}
\noindent where $k = \lceil \frac{j}{2} \rceil$.
That is, $\delta_j(x,y)$ is the formula obtained by restricting $\delta(x)$ to the first $j$ formulas of $\psi_1, \varphi_1, \psi_2, \varphi_2, \ldots, \psi_n$.  Now consider the PQ2MLO formula:
 %$\delta'(x):=\theta_1(x) \wedge \fDia^x_1 y.\theta_2(y)$ where
 \[\delta'(x)\quad =\quad \forall_x^{x+1}\ u .\bigvee_{i=1}^{2n-1} \delta_i(x,u)\: \wedge \:  \fDia^x_1 y .\exists_{y-1}^y u .\delta(u,y).
\]
 
% \[\theta_1(z) = \forall_z^{z+1}\ u .\bigvee_{i=1}^{2n-1} \delta_i(z,u)\qquad\textrm{and}\qquad
%\theta_2(z) = \exists_{z-1}^z u .\delta(u,z).
%\]%\end{eqnarray*}

\noindent The following result completes the proof of Theorem~\ref{thm:main2plus} and hence Theorem~\ref{thm:main2}.
\begin{lemma} 
$\delta(x)$ is equivalent to $\delta'(x)$.
\end{lemma}
\begin{proof}
$\delta(x) \Rightarrow \delta'(x)$.  Let $x_0, \ldots, x_{n} \in [x,x+1]$ be witnesses for the existential quantifiers in $\delta$.  From the definition of $\delta_i$, if $u \in (x_i,x_{i+1})$ (for $0 \leq i < n$) then $\delta_{2i+1}(x,u)$ holds.  Further, if $u = x_i$ (for $1 \leq i < n$) then $\delta_{2i}(x,u)$ holds.  Thus the first conjunct of $\delta'$ is satisfied for all $u \in (x,x+1)$.  Any $u \in (x,x_1)$ is a witness for $\exists_{x}^{x+1} u .\delta(u,x+1)$, and as $x_1\leq x+1$, $u \in (y-1,y)$ where $y=x+1$.  Thus the second conjunct holds and $\delta'(z)$ is satisfied.

\vspace*{1ex}

$\delta'(x) \Rightarrow \delta(x)$.  Note that if $\delta_r(x,u)$ holds for $u$ arbitrarily close to $x+1$ then $\delta_r(x,x+1)$ holds.  In particular, if $\delta_{2n-1}(x,u)$ holds for $u$ arbitrarily close to $x+1$ then we are done.  As $\bigvee_{i=1}^{2n-1} \delta_i(x,u)$ holds for all $u \in (x,x+1)$, there is some $r$ such that $\delta_r(x,u)$ holds arbitrarily close to $x+1$.  It follows that $\delta_r(x,x+1)$ is satisfied.  Suppose $r<2n-1$, and let $x_0 = x, x_1, \ldots, x_k=x+1$ be witnesses for the existential quantifiers in $\delta_r(x,x+1)$.  For convenience let $x_j=x+1$ for $k < j < n$.  Note that $k=\lceil\frac{r}{2}\rceil \leq n-1$, so it is always the case that $x_{n-1}=x+1$.  

From the second conjunct of $\delta'$, $\delta(u,x+1)$ is satisfied for some $u\in(x,x+1)$.  Let $x_0', \ldots, x_n' \in [x,x+1]$ be the witnesses for $\delta(u,x+1)$.  Let $m$ be the smallest index such that $x_m'<x_m$.  As $x_{n-1}'<x+1=x_{n-1}$ such an index must exist.  Then we claim that $x,x_1, \ldots, x_{m-1}, x_m', x_{m+1}', \ldots, x_{n-1}',x+1$ are witnesses for $\delta(x)$.  Every interval $I$ defined by these witnesses,%\footnote{$I_{2k} = \{z_k\}$ and $I_{2k+1} = (z_{k},z_{k+1})$ where $z_i = x_i$ if $i<m$ and $x_i'$ if $i\geq m$}, 
except $(x_{m-1},x_m')$, is either an interval defined by witnesses of $\delta_r(x,x+1)$ or an interval defined by witnesses of $\delta(u,x+1)$, so all points in $I$ satisfy $\psi_i$ or $\varphi_i$ as required.  For the remaining interval, we observe that $(x_{m-1},x_m') \subseteq (x_{m-1},x_m)$, thus all points satisfy $\psi_{m-1}$ as required.  Thus $\delta(x)$ is satisfied.
\end{proof}

%%%%%%%%%%%%%%%%%%%%%%%%%
%\section{Conclusion and further work}\label{sec:conc}

%%%%%%%%%%%%%%%%%%%%%%%%%

\newpage
\appendix
\section{Proof of Lemma~\ref{lem:nondense} and ``only if'' of Theorem~\ref{thm:main1}}

To prove Lemma~\ref{lem:nondense} we first need to show that $\KKK$ does not contain any limit points if it is not dense.
\begin{lemma}\label{lem:discrete}
If $\KKK$ is not dense then there exists $\epsilon>0$ such that $|B_{\epsilon}(x)|=1$ for all $x \in \KKK$.
\end{lemma}
\begin{proof}
Suppose for all $\epsilon>0$ there exist $x,y \in \KKK)$ such that $0<x-y<\epsilon$.  Take any $a<b \in \KKK$.  Then there exists $x,y \in \KKK$ such that $0 < x-y < b-a$.
That is,
\[a < a+x-y < b.\]
However, as $a,x,y \in \KKK$, it follows that $a+x-y \in \KKK$.  Thus $\KKK$ is dense.
\end{proof}

\begin{lemmaNo}{\ref{lem:nondense}}
Let $\KKK$ be an additive subgroup of $\RR$.  If $\KKK$ is not dense then $\KKK = \epsilon \ZZ$ for some $\epsilon > 0$.
\end{lemmaNo}
\begin{proof}
As $\KKK$ is not dense, it contains $\alpha \neq 0$, and without loss of generality, we may assume $\alpha>0$.  From Lemma~\ref{lem:discrete}, there exists $\delta>0$ such that  $|B_{\delta}(x)|=1$ for all $x \in \KKK$.  It follows that 
\[1 \leq |\KKK \cap (0,\alpha]| \leq \frac{\alpha}{\delta}.\]
Let $\epsilon$ be the smallest element of $\KKK \cap (0,\alpha]$, so $\epsilon$ is the smallest positive element of $\KKK$.  

We claim $\KKK = \epsilon \mathbb{Z}$.  Clearly as $\epsilon \in \KKK$, $\epsilon \mathbb{Z} \subseteq \KKK$.
Now take any $x \in \KKK$ and let $n = \lfloor \frac{x}{\epsilon}\rfloor$ and $\beta = x-n\cdot\epsilon$.  From the definition of $n$, $0 \leq \beta < \epsilon$.  As $x,\epsilon \in \KKK$ and $n \in \mathbb{Z}$, it follows that $\beta \in \KKK$.  From the definition of $\epsilon$, it follows that $\beta = 0$, so $x = n\cdot \epsilon$.  Thus $\epsilon \mathbb{Z} \supseteq \KKK$.
\end{proof}

Now suppose $\KKK$ is not dense (so $\KKK = \epsilon \ZZ$), but $\mtl_\KKK = \fo_\KKK$, so 
for any formula $\varphi^\epsilon(x) \in \fo_\KKK$ there exists a formula $\psi^\epsilon \in \mtl_\KKK$ such that for every signal $f:\domain \to 2^{\boldsymbol{P}}$ and $r \in \RR$ we have
\[ f \models \varphi^\epsilon[r] \quad \Longleftrightarrow\quad f, r \models \psi^\epsilon.\]
Consider any formula with one free variable 
$\varphi(x) \in \fo_\ZZ$.  Let $\varphi^\epsilon(x) \in \fo_{\epsilon \ZZ}$ be obtained by replacing each constant $c$ in $\varphi$ by $\epsilon c$.  It is 
clear that for every signal $f:\domain \to 2^{\boldsymbol{P}}$ and $r \in \RR$ we have
\[ f \models \varphi[r] \quad \Longleftrightarrow\quad f^\epsilon \models \varphi^\epsilon[\epsilon r]\]
where the signal $f^\epsilon$ is defined by $f^\epsilon(x) = f(\frac{x}{\epsilon})$.  
Dually, given any $\mtl_\KKK$ formula $\psi^\epsilon$, there exists an $\mtl_\ZZ$ formula $\psi$, given by replacing each interval endpoint $c$ in $\psi^\epsilon$ by $\frac{c}{\epsilon}$, such that
 for every signal $f:\domain \to 2^{\boldsymbol{P}}$ and $r \in \RR$ we have
\[ f,r \models \psi^\epsilon \quad \Longleftrightarrow\quad f^{1/\epsilon}, \frac{r}{\epsilon} \models \psi.\]
Thus for any formula $\varphi(x) \in \fo_{\ZZ}$ there exists a formula $\psi \in \mtl_\ZZ$ such that
\[ f \models \varphi[r] \quad \Longleftrightarrow\quad (f^\epsilon)^{1/\epsilon}, \frac{\epsilon r}{\epsilon} \models \psi(x).\]
But $(f^\epsilon)^{1/\epsilon} = f$ and $\frac{\epsilon r}{\epsilon} = r$, so $\psi$ is equivalent to $\varphi$.  Thus $\fo_\ZZ = \mtl_\ZZ$ which is a contradiction.

%%%%%%%%%%%%%%%%%%%%%%%%%%%%%%%%%%%
\section{Proof of Lemma~\ref{lem:HIF2}}
The following property of HIFs will prove useful:
\begin{lemma}\label{lem:HIF}
Let $\varphi(\overline{x}) = \exists y \in (s,t).\psi(\overline{x},y)$ be a HIF.  Then $\varphi(\overline{x})$ is equivalent to $u \in (s,t) \rightarrow \left(\theta^<(\overline{x}) \vee \theta^=(\overline{x}) \vee \theta^>(\overline{x})\right)$ where $u$ is a term of $\overline{x}$, $\theta^<(\overline{x}) = \exists y \in (s,u) \psi^<(\overline{x},y)$ and $\theta^<(\overline{x}) = \exists y \in (u,t) \psi^>(\overline{x},y)$ are HIFs and $\theta^=(\overline{x})$ is a HIF with strictly smaller quantifier depth than $\varphi$.
\end{lemma}
\begin{proof}
Clearly $\varphi(\overline{x})$ is equivalent to $u \in (s,t) \rightarrow \big(\exists y \in (s,u).\psi(\overline{x},y) \vee \psi(\overline{x},u) \vee \exists y \in (u,t).\psi(\overline{x},y)\big)$.  As $\psi$ has strictly smaller quantifier depth than $\varphi$, defining  $\theta^=(\overline{x}) = \psi(\overline{x},u)$ suffices.  We focus on the first disjunct to define $\theta^<$, the definition of $\theta^>$ from the third disjunct is analogous.  We proceed by induction on the quantifier depth of $\psi$.  If $\psi$ is quantifier-free then it is a HIF so set $\theta^< = \exists y \in (s,u).\psi$.  The only interesting inductive case is if $\psi = \exists z \in (y+c,t+c).\chi(\overline{x},y,z)$ for some $c \in \KKK$, the case $\psi = \forall z \in (y+c,t+c).\chi(\overline{x},y,z)$ is handled similarly.  Applying the induction hypothesis (recall $y \in (s,u)$ so $u+c \in (y+c,t+c)$) we can have that $\psi$ is equivalent to a disjunction of HIFs $\psi^< = \exists z \in (y+c,u+c) \eta^< \vee \eta^= \vee \exists z \in (u+c,t+c) \eta^>$ and it follows that $\theta^< = \exists y \in (s,u) \psi^<$ is a HIF.
\end{proof}

We now show that every $N$-bounded $\fo_\KKK$ formula with one free variable is equivalent to a HIF.  We start with a more general statement.
\begin{lemma}
Every $\fo_\KKK$ formula $\psi(\overline{x}) \in \mathrm{Bet}_{\KKK}(x_0-N,x_0+N)$ is equivalent to a disjunction $\bigvee_i (\kappa_i(\overline{x}) \wedge \varphi_i(\overline{x})$ where each $\kappa_i$ is a conjunction of constraints of the form $x_j + c < x_k + c'$ and each $\varphi_i$ is a HIF.
\end{lemma}
\begin{proof} 
We prove this by induction on the quantifier depth of $\psi$.  We can remove the equality predicate by substitution (and induction on the number of variables), so for simplicity we assume that all inequalities are strict and occur within the scope of an even number of negations.  In particular, we see that if the result holds for $\psi$ then it also holds for $\neg \psi$ as negations of inequality constraints are also inequality constraints and negations of HIFs are also HIFs.  Now if $\psi$ is quantifier-free the result follows by taking a disjunctive normal form of $\psi$.  So suppose $\psi = \exists y \varphi(\overline{x},y)$.  By the induction hypothesis we have $\varphi(\overline{x},y)$ is equivalent to $\bigvee_i (\kappa_i(\overline{x},y)\wedge \varphi_i(\overline{x},y))$, so $\psi$ is equivalent to 
\[ \bigvee_i (\kappa'(\overline{x} \wedge \exists y \bigwedge_{j=0}^n y < x_j+c_j \wedge \bigwedge_{j=0}^n y > x_j + c_j' \wedge \varphi_i(\overline{x},y)).\]
For technical reasons that will become clear shortly, we need to remove from each $\varphi_i$ intervals of the form $(y+c, y+c')$.  To do this, we observe that, by the pigeon-hole principle, $x_0 + n(c'-c) \in (y+c,y+c')$ for some $n \in \ZZ$.  As $\psi \in \mathrm{Bet}_{\KKK}(x_0-N,x_0+N)$ we have $c-N < n(c'-c) < c'+N$, so there are a finite number of possibilities for $n$, and as $c,c' \in \KKK$, $n(c'-c) \in \KKK$.  Thus for each interval $I = (y+c,y+c')$ occurring in $\varphi_i$ we take a disjunction over all integers $n$ in $(c-N,c'+n)$, add the constraints $y+c < x_0 + n(c-c') < y+c'$, and use Lemma~\ref{lem:HIF} to remove $I$.  We also assume that all constraints amongst $\overline{x}$ and $y$ implicitly defined\footnote{For example $\exists z \in (x,y)$ implicitly implies $x<y$} by $\varphi_i$ are included in the conjunction of inequalities.

The idea is to now take a disjunction over all possible choices for the greatest lower bound, $x_l+c_l$, and the least upper bound, $x_r+c_r'$, for $y$.  This adds some additional constraints (e.g.\ $x_l + c_l > x_j + c_j$ for all $j \neq l$) which we add to $\kappa'$ in each disjunct.  Now $\psi$ is equivalent to
\[ \bigvee_{i'} (\kappa''(\overline{x} \wedge \exists y \in (x_l + c_l, x_r + c_r')\, \varphi_i(\overline{x},y)).\]
We next apply Lemma~\ref{lem:HIF} to transform $\exists y \in (x_l + c_l, x_r + c_r')\, \varphi_i(\overline{x},y)$ into a HIF.  Technically we apply it several times, once for each interval defined by free variables bounded above by $y+c$ and not bounded below by $x_l+c_l+c$ and once for each interval defined by free variables bounded below by $y+c$ and not bounded above by $x_r+c_r'$.  The assumptions that there is no interval of the form $(y+c,y+c')$ and that all constraints implicitly defined by $\varphi_i$ are included in $\kappa_i$ together with the additional constraints imposed by the choice of $x_l$ and $x_r$ guarantee that $x_l+c_l+c$ is an element of any interval bounded above by $y+c$ and $x_r+c_r'+c$ is an element of any interval bounded below by $y+c$.  Thus Lemma~\ref{lem:HIF} guarantees that in the resulting HIF, $\varphi'_i$, all intervals involving $y$ and some free variable are either of the form $(x_l+c_l+c,y+c)$ or $(y+c, x_r+c_r'+c)$.  Thus $\exists y \in (x_l + c_l, x_r + c_r')\, \varphi_i'(\overline{x},y)$ is a HIF.
\end{proof}

Lemma~\ref{lem:HIF2} now follows as a corollary as inequality constraints over one variable can be trivially resolved.

\begin{corollary}
Every $N$-bounded $\fo_\KKK$ formula with one free variable is equivalent to a HIF.
\end{corollary}

%%%%%%%%%%%%%
\section{Proof of Lemma~\ref{lem:decomp}}
\begin{lemmaNo}{\ref{lem:decomp}}
Any decomposition formula $\delta(x,x+c)$ is equivalent to an $\mtl_\KKK$ formula.
\end{lemmaNo}
\begin{proof}
We proceed by induction
on the number $n$ of existential quantifiers in $\delta(x,x+c)$.

\subsubsection*{Base case}
Let $\delta(x,x+c) = \forall u\, (x<u<x+c \rightarrow
\psi(u))$, where $\psi$ is an $\ltl$ formula.  Clearly the $\mtl_\KKK$ formula
$\fBox_{(0,c)} \psi$ is equivalent to $\delta(x,x+c)$.

\subsubsection*{Inductive case}
Let $\delta(x,x+c)$ have the form
 \begin{align*}
&\exists z_0 \ldots \exists z_n\,(x=z_0< \cdots < z_n=x+c) \\
&\quad \wedge \bigwedge\{ \varphi_i(z_i) : 0 < i < n\}\\
&\quad \wedge \bigwedge\{ \forall u\,((z_{i-1}<u<z_i) \rightarrow \psi_i(u)) :
0 < i \leq n\} \, .
\end{align*}

The idea is
to define $\mtl_\KKK$ formulas $\alpha_k,\beta_k$, $0\leq k\leq 2n$, whose
disjunction is equivalent to $\delta(x,x+c)$.  The definition of these
formulas is based on a case analysis of the values of the
existentially quantified variables $z_1,\ldots,z_{n-1}$ in $\delta$.
Let $\nu \in \KKK$ be such that $0 < \nu \leq \frac{c}{2n}$.  As 
$\KKK$ is dense and non-trivial, such an element exists.
Consider the following $r = \lceil \frac{c}{\nu}\rceil \geq 2n$ subintervals of
$(x,x+c)$: $I_0 = (x,x+\nu), I_1 = [x+\nu,x+2\nu),\ldots,
I_{r-1} = [x+(\lceil\frac{c}{\nu}\rceil - 1) \nu,x+c)$.  
For simplicity we will assume $I_k = [x+k\nu,x+(k+1)\nu)$, the special instances of $I_0$ and $I_{r-1}$ where this is not the case are easily handled.

We identify three cases
according to the distribution of the $z_i$ among these intervals:
\begin{enumerate}
\item $\{z_1,\ldots,z_{n-1}\} \subseteq I_k$ for
  some $k<\frac{r}{2}$;
\item $\{z_1,\ldots,z_{n-1}\} \subseteq I_k$ for some $k$, $\frac{r}{2} \leq k < r$;
\item There exists $k$ and $l$, $1\leq l < n-1$, such
  that $z_l < x+k\nu \leq z_{l+1}$ (i.e., $z_1,\ldots,z_{n-1}$
  are not all contained in a single interval).
\end{enumerate}

\paragraph{Case 1.} 
Assume that $k<\frac{r}{2}$ and consider the following $\mtl_\KKK$ formula:
\[
\renewcommand{\arraystretch}{1.5}
\begin{array}{llll}
%\label{eq:formula}
\alpha_k \, :=&\psi_1 \mathrel{\mathbf{U}_{[k\nu,(k+1)\nu)}}  
                 \\&\quad( \varphi_1 \wedge 
               ( \psi_2\mathrel{\mathbf{U}_{(0,\nu)}}\\
              &\quad\quad( \varphi_2 \wedge 
               ( \psi_3 \mathrel{\mathbf{U}_{(0,\nu)}}\\
              & \quad\qquad\qquad \ddots\\
              &\quad\qquad( \varphi_{n-2} \wedge
               ( \psi_{n-1} \mathrel{\mathbf{U}_{(0,\nu)}}\\
              &\quad\qquad\quad( \varphi_{n-1} \wedge\fBox_{(0,\nu)} \psi_n)) \cdots )
              \\&\wedge \quad\fBox_{((k+1)\nu,c)} \psi_n \, .\notag
\end{array}\]
By construction, if $\alpha_k$ holds at a point $x$ then the formulas
$\psi_1,\varphi_1,\ldots,\varphi_{n-1},\psi_n$ hold in sequence
along the interval $(x,x+c)$.  In particular, $\psi_n$ holds on the
interval starting at the time that the subformula
$\fBox_{(0,\nu)} \psi_n$ begins to hold and extending to time
$x+c$ (thanks to the ``overlapping'' subformula
$\fBox_{((k+1)\nu,c)} \psi_n$).  Thus $\alpha_k$ implies
$\delta(x,x+c)$.  Conversely, if $\delta(x,x+c)$ holds with the existentially
quantified variables $z_1,\ldots,z_{n-1}$ all lying in the interval
$[x+k\nu,x+(k+1)\nu)$, then clearly $\alpha_k$ also
holds.

\paragraph{Case 2.} 

Suppose that $\frac{r}{2} \leq k < r$ and consider the following $\mtl_\KKK$ formula:
\[
\renewcommand{\arraystretch}{1.5}
\begin{array}{lll}
\alpha_k \, := & \fDia_{\{c\}}  \big[ \psi_n \mathrel{\mathbf{S}_{(c-(k+1)\nu, c-k\nu)}}\\
&\qquad ( \varphi_{n-1} \wedge  ( \psi_{n-1} \mathrel{\mathbf{S}_{(0,\nu)}}\\
&\qquad\quad( \varphi_{n-2} \wedge ( \psi_{n-2} \mathrel{\mathbf{S}_{(0,\nu)}}\\
&\qquad\qquad\qquad\quad         \ddots \\
&\qquad\qquad\quad( \varphi_{2} \wedge  ( \psi_{2} \mathrel{\mathbf{S}_{(0,\nu)}} \\
&\qquad\qquad\qquad( \varphi_{1} \wedge \pBox_{(0,\nu)} \psi_1)) \cdots )\big ]\\ 
&\wedge \quad \fBox_{(0,k\nu)} \psi_1\, .
\end{array}\]

The definition of $\alpha_k$  is according to similar principles as in
Case 1.  If it holds at a point $x$ then the sequence of past
operators ensures that the formulas
$\psi_n,\varphi_{n-1},\psi_{n-1},\ldots,\varphi_1,\psi_1$ hold in
sequence, backward from $x+c$ to $x$.  Thus $\alpha_k$ implies
$\delta(x,x+c)$.  Conversely, if $\delta(x,x+c)$ holds with the existentially
quantified variables $z_1,\ldots,z_{n-1}$ all lying in the interval
$[x+k\nu,x+(k+1)\nu)$, $\frac{r}{2} \leq k < r$, then clearly
$\alpha_k$ also holds.

\paragraph{Case 3.}  
Suppose that $z_l < x+k\nu \leq z_{l+1}$ for some $k$ and $l$, $1 \leq l < n-1$.  

The idea is, for each choice of $l$, to decompose $\delta(x,x+c)$ into a property $\sigma_l$ holding
on the interval $(x,x+k\nu)$ and a property $\tau_l$ holding
on the interval $(x+k\nu,x+c)$.  We then apply the induction
hypothesis to transform $\sigma_l$ and $\tau_l$ to equivalent $\mtl_\KKK$
formulas.  To this end, define
\begin{align*}
\sigma_l(x) \, := \, & \exists z_0 \ldots \exists z_{l+1}
                  (x=z_0 < \cdots < z_{l+1} = x+k\nu)\\
     & \wedge \bigwedge \{ \varphi_i(z_i) : 0 < i \leq l \}\\
     & \wedge \bigwedge \{ \forall u((z_{i-1}<u<z_i) \rightarrow \psi_i(u)) : 
                           1 \leq i \leq l+1 \} 
\end{align*}
and
\begin{align*}
\tau_l(x) \, := \, & \exists z_l \ldots \exists z_n
                  (x=z_l < \cdots < z_n = x+c-k\nu)\\
     & \wedge \bigwedge \{ \varphi_i(z_i) : l < i < n \}\\
     & \wedge \bigwedge \{ \forall u((z_{i-1}<u<z_i) \rightarrow \psi_i(u)) : 
                           l < i \leq n \} \, .
\end{align*}
%Note that $\tau_l$ is subtly different from $\sigma_l$ as $\varphi(z_l)$
%is not included in the second conjunction.

Now $\sigma_l \in \mathrm{Bet}_{\{0\}}(x,x+k\nu)$ and $\tau_l \in \mathrm{Bet}_{\{0\}}(x,x+c-k\nu)$
so it follows by the induction hypothesis that $\sigma_l$ and $\tau_l$ have equivalent
$\mtl_\KKK$ formulas $\sigma_l^\dag$ and $\tau_l^\dag$ respectively.

We now define
\[ \beta_k := \bigvee_{1 \leq l < n-1} \Big( \sigma_l^\dag \wedge \fDia_{\{k\nu\}} \big((\psi_{l+1}   \wedge \tau_l^\dag) \vee (\varphi_{l+1} \wedge \tau_{l+1}^\dag)\big)\Big) \, .\]
From the definition of $\sigma_l$ it is clear that $\beta_k$ matches
$\delta(x,x+c)$ on $(x,x+k\nu)$.  For the remaining interval
  $[x+k\nu,x+c)$ we distinguish between two cases: if
    $x+k\nu < z_{l+1}$, then $\fDia_{\{k\nu\}}(\psi_{l+1}
    \wedge \tau_l^\dag)$ agrees with $\delta(x,x+c)$; and if $x+k\nu
    = z_{l+1}$ then $\fDia_{\{k\nu\}}(\varphi_{l+1} \wedge
    \tau_{l+1}^\dag)$ agrees with $\delta(x,x+c)$.  Thus $\beta_k$ implies
    $\delta(x,x+c)$.  Conversely if $\delta(x,x+c)$ holds with the existentially
    variables $z_1,\ldots,z_{n-1}$ satisfying the conditions of Case 3
    then one of the disjuncts, and hence $\beta_k$, must hold.
\end{proof}

\end{document}